\documentclass[runningheads]{llncs}

\usepackage{tcolorbox}
\usepackage{enumerate}
\usepackage{comment}

\usepackage{tcolorbox}
\usepackage{amsmath}
\usepackage{enumerate}
\usepackage{graphicx}
\usepackage{amsmath}
\usepackage{amssymb}
\usepackage{comment}
\usepackage{hyperref}
\usepackage[english]{babel}

\usepackage{tikz}
\usetikzlibrary{shapes.geometric}
\usepackage{float}
\usepackage{tcolorbox}
\usepackage{graphicx}

\usepackage{import}

\def \calf {{\cal F}}

\def \dcut {{$d$-\textsc{Cut}}}
\def \kdcut {{$(k,d)$-\textsc{Cut}}}
\def \matchingcut {{\textsc{Matching Cut}}}

\begin{document}
\title{An FPT algorithm for Matching Cut and d-Cut}
%
%
\author{N R Aravind\inst{} \and
Roopam Saxena\inst{1}}
\authorrunning{N R Aravind and Roopam Saxena}
%
\institute{Department of Computer Science and Engineering\\IIT Hyderabad, Hyderabad. India\\
\email{\{aravind,cs18resch11004\}@iith.ac.in,}}
\maketitle              
\begin{abstract}

Given a positive integer $d$, the {\dcut} problem is to decide if an undirected graph $G=(V,E)$ has a cut $(A,B)$ such that every vertex in $A$ (resp. $B$) has at most $d$ neighbors in $B$ (resp. $A$). For $d=1$, the problem is referred to as {\matchingcut}. Gomes and Sau, in IPEC 2019, gave the first fixed parameter tractable algorithm for {\dcut} parameterized by maximum number of the crossing edges in the cut (i.e. the size of edge cut).
However, their paper doesn't provide an explicit bound on the running time, as it indirectly relies on a MSOL formulation and Courcelle's Theorem.
Motivated by this, we design and present an FPT algorithm for {\dcut} for general graphs with running time $2^{O(k\log k)}n^{O(1)}$ where $k$ is the maximum size of the edge cut. This is the first FPT algorithm for the {\dcut} and {\matchingcut} with an explicit dependence on this parameter. We also observe that there is no algorithm solving {\matchingcut} in time $2^{o(k)}n^{O(1)}$ where $k$ is the maximum size of the edge cut unless ETH fails.

\end{abstract}
\section{Introduction}

Given a graph $G=(V,E)$, a cut $(A,B)$ is a matching cut if every vertex in $A$ (resp. $B$) has at most $1$ neighbor in $B$ (resp $A$); equivalently a cut $(A,B)$ is matching cut if edge cut $E(A,B)$ is a matching. Note that as per these definitions matching cut can be empty, and not every matching whose removal disconnects a graph needs to be a matching cut. {\matchingcut} is the problem of deciding if a given undirected graph $G$ admits a matching cut or not.
\paragraph{}

{\matchingcut} has been extensively studied. Graphs admitting matching cuts were first discussed by Graham \cite{graham1970primitive}. Chv{\'{a}}tal \cite{Chvatal84} proved {\matchingcut} to be NP-Complete for graphs with maximum degree 4. Bonsma \cite{DBLP:journals/jgt/Bonsma09} proved {\matchingcut} to be NP-complete for planar graphs with maximum degree 4 and with girth 5. Kratsch and Le \cite{DBLP:journals/tcs/KratschL16} provided an exact algorithm with running time $O^*(1.414..^n)\footnote{We use $O^*$ notation which suppresses polynomial factors. }$ and also provided a single exponential algorithm parameterised by the vertex cover number. Komusiewicz, Kratsch and Le \cite{DBLP:journals/dam/KomusiewiczKL20} further improved the running time of branching based exact algorithm to $O^*(1.3803^n)$ and also provided a SAT based $O^*(1.3071^n)$-time randomized algorithm.  They also provided single exponential algorithm parameterized by distance to cluster and distance to co cluster. Aravind, Kalyanasundaram and Kare \cite{DBLP:conf/cocoa/AravindKK17} provided fixed parameter tractable algorithms for various structural parameters including treewidth. Recently hardness and polynomial time solvable results have been obtained for various structural assumptions in \cite{DBLP:journals/tcs/LeL19,DBLP:conf/cocoon/HsiehLLP19}.

\paragraph{}
Recently Gomes and Sau \cite{DBLP:journals/algorithmica/GomesS21} considered
a generalization of matching cut and called it $d$-cut. For a positive integer $d\geq 1$, a cut $(A,B)$ is $d$-cut if every vertex in $A$ (resp. $B$) has at most $d$ neighbor in $B$ (resp $A$).  They called {\dcut} the problem of deciding if a given graph $G$ has a $d$-cut or not. They showed that for every $d\geq 1$, {\dcut} is NP-hard for regular graphs even when restricted to $(2d+2)$-regular graphs \cite{DBLP:journals/algorithmica/GomesS21}. They considered various structural parameters to study {\dcut} and provided FPT results. They also showed fixed parameter tractability of {\dcut} when parameterized by maximum size of the edge cut (i.e. number of edges crossing the cut) using a reduction on results provided by Marx, O'Sullivan and Razgon \cite{DBLP:conf/stacs/MarxOR10}. 
However, they didn't provide an explicit bound on the running time as the treewidth reduction technique of \cite{DBLP:conf/stacs/MarxOR10} relies on MSOL formulation and Courcelle's theorem \cite{DBLP:journals/iandc/Courcelle90} to show fixed parameter tractability. Marx et. al. \cite{DBLP:conf/stacs/MarxOR10} also observed that their method may actually increase the treewidth of the graph, however the treewidth will remain $f(k)$ for some function $f$. This motivated us to investigate an FPT algorithm for {\dcut} parameterized by maximum size of the edge cut where we can explicitly bound the dependence on parameter. Note that when $d=1$, the problem can be referred to as {\matchingcut}.
\paragraph{}

We now define a parameterized version of {\dcut} with maximum size of the edge cut as parameter. We also note that we will consider $d$ to be a part of the input in our problem definition.  

\begin{tcolorbox}[colback=white]
{
{\kdcut}: \newline
\textit{Input:} An instance $I$ = $(G,k,d)$. Where graph $G=(V,E)$, $|V| = n$ and $k,d\in \mathbb{N}$.\newline
\textit{Parameter}: k.\newline
\textit{Output:} \textit{yes} if $G$ contains a $d$-cut $(A,B)$ such that $|E(A,B)|\leq k$, \textit{no} otherwise. \newline
}
\end{tcolorbox}

\

\subsection{Our Contribution}

Our main contribution is the following theorem.

\begin{theorem}\label{FPTMC}
   There exists an algorithm which runs in time $2^{O(k\log k)}n^{O(1)}$ and solves {\kdcut}.
\end{theorem}
\paragraph{}

 Cygan, Komosa, Lokshtanov, Pilipczuk, Pilipczuk, Saurabh and Wahlstr{\"{o}m} \cite{DBLP:journals/talg/CyganKLPPSW21} provided a compact tree decomposition with bounded adhesion and guaranteed unbreakability along with a framework to design FPT algorithms and showed its application on Minimum Bisection and other problems. We use their framework and compact tree decomposition with bounded adhesion and guaranteed unbreakability along with {\kdcut} specific calculations and proofs to design an FPT algorithm for {\kdcut} for the proof of Theorem \ref{FPTMC}.  
\paragraph{}

 Kratsch and Le (section 3.3 in \cite{DBLP:journals/tcs/KratschL16}) have shown that the reduction given in \cite{DBLP:conf/wg/PatrignaniP01} also implies that for an $n$ vertex graph, {\matchingcut} cannot be solved in $2^{o(n)}n^{O(1)}$ unless exponential time hypothesis (ETH) fails. For every matching cut, the maximum size of the edge cut $k$ is linearly bounded by the number of vertices in the graph. Thus, this result of Kratsch and Le \cite{DBLP:journals/tcs/KratschL16} implies the following.
\begin{corollary}
    Unless ETH fails, the problem of deciding if a given $n$ vertex graph has a matching cut with edge cut size at most $k$ cannot be solved in $2^{o(k)}n^{O(1)}$.
\end{corollary}\label{MCLB}

\section{Preliminaries}
\subsection{Multiset Notations} Considering a set $U$ as universe, a  multiset is a $2$-tuple $P =(U,m_P)$ where  multiplicity function $m_P : U \rightarrow \mathbb{Z}_{\geq 0} $ is a mapping from $U$ to non negative integers such that for an element $e\in U$, the value $m_P(e)$ is the multiplicity of $e$ in $P$ i.e the number of occurrence of $e$ in $P$. Cardinality of a multiset $P$ is the sum of multiplicity of all the distinct elements of $P$. We write $e\in P$ if $m_P(e)\geq 1$. $P$ is \textit{empty} i.e. $P=\emptyset$ if and only if $\forall e\in U$, $m_P(e) =0$. For two multiset $A$ and $B$ on universe $U$, let $m_A$ and $m_B$ be their respective multiplicity functions. We will use following operations on multisets for our purposes.
\paragraph{}
\textbf{Equality:} $A$ is equal to $B$ denoted by $A=B$, if $\forall e\in U$, $m_A(e) = m_B(e)$. We say $A$ and $B$ are \textit{distinct} iff they are not equal.
\paragraph{}
\textbf{Inclusion:} $A$ is included in $B$ denoted by $A\subseteq B$, if $\forall e\in U$, $m_A(e)\leq m_B(e)$.
\paragraph{}
\textbf{Sum Union:}  $P$ is a sum union of $A$ and $B$ is denoted by $P = A \uplus B$, let $m_P$ be the multiplicity function for $P$, then $\forall e\in U$, $m_P(e) = m_A(e) + m_B(e) $.
\paragraph{}
Throughout this paper, if the context is clear, for any multiset $X$ we will use $m_X$ to denote the multipicity function of $X$.

\subsection{Notations and Terminologies}
All the graphs consider in this paper are simple and finite. We use standard graph notations and terminologies and refer the reader to \cite{DBLP:books/daglib/0030488} for basic graph notations and terminologies. We mention some notations used in this paper. A graph  $G=(V,E)$ has a vertex set $V$ and edge set $E$. We also use $V(G)$ and $E(G)$ to denote the vertex set and edge set of $G$ respectively. 
For an edge set $F\subseteq E(G)$, $V(F)$ denotes the set of all the vertices of $G$ with at least one edge in $F$ incident on it. For two disjoint vertex sets $A,B\subseteq V$, we use $E_G(A,B)$ to denote all the edges in graph $G$ with one endpoint in $A$ and other in $B$, if the graph in context is clear then we simply use $E(A,B)$. For a vertex set $S\subseteq V(G)$, $G[S]$ denotes the induced sub graph of $G$ on vertex set $S$, and $G-S$ denotes the graph $G[V(G)\setminus S]$. For an edge set $E'\subseteq E$, $G[E']$ denotes the sub graph of $G$ on edge set $E'$ i.e. $G[E']=(V(E'),E')$. 
For an edge set $A\subseteq E(G)$, $G-A$ denotes the graph with vertex set $V(G)$ and edge set $E(G)\setminus A$.
 A component $C$ of a graph $G$ is a maximally connected subgraph of $G$.
 A graph $G$ is a forest if every component of $G$ is a tree.
 We use $G'\subseteq G$ to denote that $G'$ is a subgraph of $G$. For a $v\in V(G)$, we use $N_G(v)$ to denote the open neighborhood of $v$ in $G$ that is the set of all adjacent vertices of $v$ in $G$; and for a $v\in V(G)$, we use $N_G[v]$ to denote the closed neighborhood of $v$ in $G$ that is $N_G[v]= N_G(v)\cup \{v\}$. If the context of the graph is clear we will simply use $N(v)$ and $N[v]$. Further, for a set $Z\subseteq V$, $N(Z)= \bigcup_{v\in Z} N(v)$, similarly  $N[Z]= \bigcup_{v\in Z} N[v]$. 
 \paragraph{}
    For a graph $G=(V,E)$, a \textit{partition of $G$} refers to a partition of its vertex set $V$. We call a partition of graph trivial if the partition of its vertex set consist of a single set. For a graph $G$, a partition $(A,B)$ of its vertices is called a \textit{cut}. A cut $(A,B)$ is a \textit{$d$-cut} if every vertex in $A$ (resp. in $B$) has at most $d$ neighbors in $B$ (resp. in $A$). For a graph $G=(V,E)$, an edge set $M\subseteq E$ is called a \textit{$d$-matching} if every vertex $v\in V$ has at most $d$ edges in $M$ incident on it. Observe that a cut $(A,B)$ is $d$-cut if and only if $E(A,B)$ is a $d$-matching.
\paragraph{}

\subsection{Parameterized Complexity} 
For details on parameterized complexity, we refer to \cite{DBLP:books/sp/CyganFKLMPPS15,DBLP:series/txcs/DowneyF13}, and recall some definitions here.

\begin{definition}[\cite{DBLP:books/sp/CyganFKLMPPS15}]
A \textit{parameterized problem} is a language $L \subseteq \Sigma^* \times \mathbb{N} $ where $\Sigma$ is a fixed and finite alphabet. For an instance $I=(x,k) \in \Sigma^* \times \mathbb{N} $, $k$ is called the parameter. 
\end{definition}

\begin{definition}[\cite{DBLP:books/sp/CyganFKLMPPS15}]
    A parameterized problem is called \textit{fixed-parameter tractable} if there exists a computable function $f:\mathbb{N} \to \mathbb{N}$, a constant $c$, and an algorithm $\cal A$ (called a \textit{fixed-parameter algorithm} ) such that: the algorithm $\cal A$ correctly decides whether $I\in L$ in time bounded by $f(k).|I|^c$. The complexity class containing all fixed-parameter tractable problems is called FPT.
\end{definition}

\subsection{Tree decomposition} 
\textbf{Tree Decomposition} \cite{DBLP:books/sp/CyganFKLMPPS15} : A tree decomposition of a graph $G$ is a pair $(T,\beta)$ where $T$ is a tree and $\beta$ (called a bag) is a mapping that assigns to every $t \in V (T)$ a set $\beta(t)\subseteq V (G)$, such that the following
holds: 
\begin{enumerate}
    \item For every $e \in E(G)$, there exists a $ t\in V(T)$ such that $V(e) \subseteq \beta (t);$
    \item For $v\in V(G)$, let $\beta ^{-1}(v)$ be the set of all vertices $t \in V(T)$ such that $v \in \beta (t)$, then  $T[\beta ^{-1}(v)]$ is a connected nonempty subgraph of $T$.
\end{enumerate}
If the tree $T$ is rooted at some node $r$, we call it a \textit{rooted tree decomposition}.

\begin{definition}
\textbf{Adhesion in tree decomposition:} In a tree decomposition $(T,\beta)$, For an edge $e\in E(T)$ where $e= \{t_1,t_2\}$, a set $\sigma(e)$ = $\beta(t_1) \cap \beta(t_2)$ is called adhesion of $e$. For a \textit{rooted tree decomposition} $(T,\beta)$ adhesion of a node $t\in V(T)$ denoted by $\sigma(t)$ is $\sigma(\{t,t'\})$ where $t'$ is parent node of $t$ in $T$. Adhesion of a root node $r$ is $\emptyset$.
\end{definition}

We recall the following from \cite{DBLP:journals/talg/CyganKLPPSW21}.
\paragraph{}
\textbf{Some functions for convenience}:
For a rooted tree decomposition $(T,\beta)$ at some node $r$, for $s,t\in V(t)$ we say that $s$ is a \textit{descendent} of $t$, if $t$ lies on the unique path from $s$ to the $r$. This implies that a node is a descendant of itself.\\
 $\gamma(t) = \bigcup _{c:\ \text{descendant of t}} \beta(c), \    \alpha(t) = \gamma(t)\setminus \sigma(t),\  G_t= G[\gamma(t)] - E(G[\sigma(t)])$.

\begin{definition}[\cite{DBLP:journals/talg/CyganKLPPSW21}]
\textbf{Compact tree decomposition:}
A rooted tree decomposition $(T,\beta)$ of $G$ is compact if for every non root-node $t\in V(T):$ $G[\alpha(t)]$ is connected and $N(\alpha(t)) = \sigma(t)$.
\end{definition}

\begin{definition}[\cite{DBLP:journals/talg/CyganKLPPSW21}]
\textbf{Separation:}
A pair $(A,B)$ of vertex subsets in a graph $G$ is a separation if $A\cup B= V(G)$ and there is no edge with one endpoint in $A\setminus B$ and the other in $B\setminus A$; the order of the separation $(A,B)$ is $|A\cap B|$.
\end{definition}

In \cite{DBLP:journals/talg/CyganKLPPSW21} the edge cut $(A,B)$ is defined as a pair $A,B \subseteq V(G)$ such that $A\cup B = V(G)$ and $A \cap B =\emptyset$, which we refer to as partition $(A,B)$. And the order of the edge cut $(A,B)$ is defined as $|E(A,B)|$. These terminologies are required for following definition.

\begin{definition} [\cite{DBLP:journals/talg/CyganKLPPSW21}]
\textbf {Unbreakability}: Let $G$ be a graph, A vertex subset $X\subseteq V(G)$ is (q,k)-unbreakable if every separation(A,B) of order at most k satisfies $|A\cap X|\leq q$\ or \ $|B\cap X|\leq q$. A vertex subset $Y\subseteq V(G)$ is (q,k)-edge-unbreakable if every edge cut (A,B) of order at most k satisfies $|A\cap Y|\leq q$ \ or\ $|B \cap Y|\leq q$.
\end{definition}
Observe that every set that is \textit{(q,k)-unbreakable} is also \textit{(q,k)-edge-unbreakable}.

\begin{theorem}[\cite{DBLP:journals/talg/CyganKLPPSW21}]\label{lean-decomp}
Given an $n$-$vertex$ graph $G$ and an integer $k$, one can in time $2^{O(k \log k)} n^{O(1)}$ compute a rooted compact tree decomposition $(T,\beta)$ of $G$ such that
\begin{enumerate}
    \item every adhesion of $(T,\beta)$ is of size at most $k$;
    \item every bag of $(T,\beta)$ is $(i,i)$-unbreakable for every $1\leq i \leq k$.
\end{enumerate}
\end{theorem}
Note that since every bag of the output decomposition $(T,\beta)$ of Theorem \ref{lean-decomp} is $(k,k)$-unbreakable, it is also $(k,k)$-edge-unbreakable. Further, the construction provided for the proof of Theorem $\ref{lean-decomp}$ in  \cite{DBLP:journals/talg/CyganKLPPSW21} maintained that the number of edges in decomposition is always upper bounded by $|V(G)|$ and hence $|V(T)|\leq |V(G)|+1$.


\subsection{Color coding tools:}
\begin{lemma}[\cite{DBLP:conf/focs/ChitnisCHPP12}] \label{color}
Given a set $U$ of size $n$, and integers $0\leq a,b \leq n$, one can in $2^{O(min(a,b)\log (a+b))}n\log n$ time construct a family $\calf$ of at most\\ $2^{O(min(a,b)\log (a+b))}\log n$ subsets of $U$, such that following holds: for any sets $A,B\subseteq U$, $A\cap B =\emptyset$, $|A|\leq a$, $|B|\leq b$, there exist a set $S\in \calf$ with $A\subseteq S$ and $B\cap S =\emptyset$

\end{lemma}


\section{FPT Algorithm}

In this section we give the proof of Theorem \ref{FPTMC}.
\paragraph{}
A disconnected graph $G$ trivially has a $d$-cut of size $0$ and thus, $(G,k,d)$ is always a yes instance of {\kdcut} for every $k,d\in \mathbb{N}$. Thus, it remains to find if $(G,k,d)$ is a yes instance when the graph $G$ is connected. From here onward we will always assume that input graph $G$ is simple and connected. Further, we can also assume that $d<k$, otherwise the problem become same as deciding if $G$ has a min cut of size at most $k$ which is polynomial time solvable \cite{stoerWagner}.

\paragraph{}
We will start by invoking Theorem \ref{lean-decomp} on input $n$ vertex graph $G$ with parameter $k$. This gives us a rooted compact tree decomposition $(T,\beta)$ of $G$ where every bag is $(k,k)$-edge-unbreakable and every node $t\in V(T)$ has adhesion of size at most $k$.

\subsection{Designing Memory Table}

\begin{definition}
\label{dcandidate_node}{\bf(Matched candidate set of a vertex set $Q\subseteq V$)}
For a vertex set $Q\subseteq V$ and $d\in \mathbb{N} $ ,\, we call a multiset $P=(V,m_P)$ a \textit{$d$-matched candidate set} of $Q$ if the following holds.
\begin{itemize}
    \item $\forall v\in Q$, $m_P(v)\leq d$,
    \item $\forall v\in V\setminus Q$, $m_P(v)=0$,
    \item $|P|\leq k$.
\end{itemize}
\end{definition}
Note that if $Q=\emptyset$, then an empty multiset $P=\emptyset$ that is $\forall v\in V$, $m_P(v)=0$ is the only $d$-matched candidate set of $Q$.  

\begin{proposition}\label{MatchedSet}
 Given $Q\subseteq V$, if $|Q|\leq k$, then there are at most $2^{O(k\log k)}$ distinct $d$-matched candidate sets of $Q$ and in time $2^{O(k\log k)}n^{O(1)}$ we can list all of them.
\end{proposition}
\begin{proof}
Consider $d$ copies of every $v\in Q$ and treat all of them as distinct elements, clearly if $|Q|\leq k$, then there are at most $kd$ such elements. 
Every $d$-matched candidate set $P$ of $Q$ can be considered as selection of $|P|$ elements from these $kd$ elements, as we can define multiplicity of each vertex $v\in Q$ to be equal to the number of selected copies of $v$ and multiplicity of each vertex $v\in V\setminus Q$ as $0$. 
Thus, number of all possible subsets of size at most $k$ of the set of these $kd$ elements is an upper bound for all possible $d$-matched candidate sets of $S$. Clearly, two selected subsets may produce equal $d$-matched candidate sets with this procedure. However, for the simplicity of the analysis we are using this procedure and the bound obtained is sufficient for our purpose. The number of possible subsets of size at most $k$ of set containing $kd$ elements are bounded by $\sum_{i=0}^k \binom{kd}{i}$ which is bounded by $2^{O(k\log k)}$ as $d$ is at most $k$. 
\qed \end{proof}
\paragraph{}
Following the framework given in \cite{DBLP:journals/talg/CyganKLPPSW21}, we perform a bottom up dynamic programming on $(T,\beta)$ by computing entries of a table $M$ and show that the computation of $M$ is sufficient to solve {\kdcut}. Formally, for every vertex $t\in V(T)$, every set $S \subseteq \sigma (t)$, $\bar{S} = \sigma(t)\setminus S$, every $d$-matched candidate set $P=(V,m_P)$ of $\sigma(t)$ and $n_e\in \{0,1\}$ we compute an integer $M[t, S, P, n_e] \in \{0,1,2,....,k,\infty\}$ which satisfies the following properties. 

\begin{enumerate}[(1)]
    \item 
If $M[t,S,P,n_e] \leq k$, then there exists a partition $(A,B)$ of $G_t$ such that the following holds.
    \begin{itemize}
        \item If $n_e=1$, then both $A$ and $B$ are non empty, otherwise either $A$ or $B$ is empty,
        \item $A \cap \sigma(t) \in \{S,\bar{S}\}$,
        \item $E_{G_t}(A,B)$ forms a $d$-matching,
        \item every vertex $v$ in $\sigma(t)$ has at most $m_P(v)$ neighbors in other side of the partition in $G_t$ i.e. $\forall v\in \sigma(t),\ |N_{G[E_{G_t}(A,B)]}(v)| \leq m_P(v)$,
        \item $|E_{G_t}(A,B)| \leq M[t, S,P,n_e]$.
        
    \end{itemize}
    
    \item For every partition $(A,B)$ of the graph $G$ that satisfies the following prerequisites:
     \begin{itemize}
        \item $A \cap \sigma(t) \in \{S,\bar{S}\}$,
        \item $E_G(A,B)$ forms a $d$-matching,
        \item every vertex $v$ in $\sigma(t)$ has at most $m_P(v)$ neighbors in other side of the partition in $G_t$ i.e. $\forall v\in \sigma(t),\ |N_{G[E_{G_t}(A,B)]}(v)| \leq m_P(v)$,
        \item $|E_G(A,B)|\leq k$.
    \end{itemize}
 It holds that $|E_{G_t}(A,B)| \geq M[t,S,P,1]$ if both $V(G_t)\cap A$ and $V(G_t)\cap B$ are non empty, otherwise $|E_{G_t}(A, B)| \geq M[t,S,P,0]$ if either $V(G_t)\cap A$ or $V(G_t)\cap B$ is empty.\\
We note that $|E_{G_t}(A,B)|\geq \infty$ imply that such partition $(A,B)$ doesn't exist. We now formally claim that computation of table $M$ is sufficient. 
\end{enumerate}

\begin{lemma}\label{table_correctness}
  $(G,k,d)$ is a yes instance of {\kdcut} if and only if $M[r,\emptyset,\emptyset,1]\leq k$ where $r$ is the root of $T$.
\end{lemma}
\begin{proof}
For the first direction, by property 1 for $M[r,\emptyset,\emptyset,1]$, a non trivial partition $(A,B)$ for $G_r=G$ exists which is a $d$-cut of $G$ with $|E_G(A,B)|\leq k$.

For the other direction, let $(A,B)$ be a $d$-cut of $G$ such that $|E_G(A,B)| \leq k$. Since $\sigma (r)= \emptyset$, $(A,B)$ satisfies all the prerequisites of property $(2)$ for $t=r$, $S=\emptyset$ and $P=\emptyset$. Further, $V(G_r)\cap A$ and $V(G_r)\cap B$ are both non empty as $(A,B)$ is a non trivial partition. Thus,  $k \geq |E_{G_r}(A,B)|\geq M[r,\emptyset,\emptyset,1]$.  This finishes the proof.  
\qed \end{proof}

\begin{proposition}\label{empty}
 For every $t$, $S$ and $P$, if either $S$ or $\bar{S}$ is empty, then\\ $M[t,S,P,0] = 0$ satisfies property $(1)$ and $(2)$. 
\end{proposition}
\begin{proof}
Consider the partition $(A=\phi, B=V(G_t))$ of $G_t$, it satisfies all the statements of property $(1)$ for $M[t,S,P,0] = 0$ for $S\in \{\emptyset,\sigma(t)\}$ and for every $d$-matched candidate set $P$ of $\sigma(t)$. Further, property $(2)$ is trivially satisfied, as for any partition $(A,B)$ of $G$, $|E_{G_t}(A,B)| \geq 0$.
\qed \end{proof}

\begin{proposition}\label{infty}
 For every $t$, $S$ and $P$, if both $S$ and $\bar{S}$ are non empty, then $M[t,S,P,0] = \infty$.
\end{proposition}
\begin{proof}
 It is straightforward to see that if both $S$ and $\bar{S}$ are non-empty, then there cannot be a trivial partition $(A,B)$ of $G_t$.
\qed \end{proof}

\begin{proposition}\label{nonempty}
For every $t$, $S$ and $P$, if both $S$ and $\bar{S}$ are non empty, then $M[t,S,P,1]\geq 1$. 
\end{proposition}
\begin{proof}
If $M[t,S,P,1]\leq k$, then consider the partition $(A,B)$ of $G_t$ corresponding to property $(1)$, as both $S$ and $\bar{S}$ being non empty, both $A$ and $B$ are non empty. Recalling the compactness of tree decomposition $T$, where $G[\alpha(t)]$ is connected and $N_G(\alpha(t)) = \sigma(t)$. Thus, for any such partition, $E_{G_t}(A,B)$ can't be empty. 
\qed \end{proof}
\paragraph{}

There are $O(n)$ nodes in $T$, hence if for a fixed node $t\in V(T)$ we can compute the entries of the table $M$ in time $2^{O(k\log k)}n^{O(1)}$, then we can prove Theorem $\ref{FPTMC}$.
For every $t\in V(T)$, $|\sigma(t)|\leq k$ and thus there are at most $2^{O(k\log k)}$ choices of $S$. Recalling proposition $\ref{MatchedSet}$ the number of distinct $d$-matched candidate sets of $\sigma(t)$ is bounded by $2^{O(k\log k)}$ and we can obtain them in time $2^{O(k\log k)}n^{O(1)}$. Thus, the number of entries in $M$ for a fixed $t\in V(T)$ are bounded by $2^{O(k\log k)}$. Thus, if we can show that a single entry $M[t,S,P,n_e]$ can be computed in time $2^{O(k\log k)}n^{O(1)}$ given that the entries of $M$ for every children $c$ of $t$ in $T$ is available, we can bound the time required to compute all the entries of $M$ for a node $t\in V(T)$ to $2^{O(k\log k)}\cdot2^{O(k\log k)}n^{O(1)}$ which is essentially $2^{O(k\log k)}n^{O(1)}$. Thus, from now onwards we focus on the computation of a single entry $M[t,S,P,n_e]$.


\subsection{Computing an Entry of Memory Table}
\label{section_M[tsp]_calculation}
In this section we will discuss the  computation of a single entry $M[t,S,P,n_e]$. For a given $t\in V(t)$, $S\subseteq \sigma(t)$ and $P$ such that $P$ is a $d$-matched candidate set of $\sigma(t)$. We use proposition \ref{empty} and \ref{infty} to set $M[t,S,P,0]\in\{0,\infty\}$. Thus, we move on to the  computation of $M[t,S,P,1]$. Let $Z(t)$ be the set of all the children of $t$ in $T$. From here on we will assume that entries of $M$ corresponding to every $c\in Z(t)$ are calculated. Note that if $t$ is a leaf vertex, then $Z(t)$ is empty.
\paragraph{}
Following the framework given in \cite{DBLP:journals/talg/CyganKLPPSW21}, in this step of the algorithm, we focus on partitioning $\beta(t)$. Intuitively, we want to find the best way to partition $G_t$, for this we will try to partition $\beta(t)$ and use entries $M[c,.,.,.]$ to know the best way to partition subgraph $G_c$. In this regard, we also consider every edge $e\in E(G_t[\beta(t)])$ as a subgraph of $G_t$ on the similar lines of $G_c$, and to find the best way to partition subgraph $G[e]$ we construct a table $M_E[e,S',P']$ for every edge $e=\{u,v\}\in E(G_t[\beta(t)])$ where $S' \subseteq \{u,v\}$, and $P'$ is a 1-matched candidate set of $\{u,v\}$. We are taking a 1-matched candidate set of $V(e)$, as there is only 1 edge in $G[e]$. 
\paragraph{}
\textbf{Construction 3.1:} We assign the following values to $M_E[e,S',P']$.
    \begin{itemize}
        \item $M_E[e,\emptyset, P']$ = $M_E[e,\{u,v\}, P']$ = $0$ for every $1$-matched candidate set $P'$ of $V(e)$;
        \item $M_E[e,\{u\},P']$ = $M_E[e,\{v\},P']$ = $1$ for $1$-matched candidate set $P'$ of $V(e)$ such that $m_{P'}(u) =m_{P'}(v) =1$;
        \item  $M_E[e,\{u \},P']$ = $M_E[e,\{v\},P']$= $\infty$ for every $1$-matched candidate set $P'$ of $V(e)$ such that $m_{P'}(u)=0$ or $m_{P'}(v) =0$.
    \end{itemize}
Intuitively, if both $u,v$ falls into the same side of the partition, 
then $M_E[e,S',P']$ costs $0$. Otherwise, if $u$ and $v$ falls into different side of the partition and both are allowed to have a neighbor in the other side of the partition in $G[e]$ as per $P'$, then $M_E[e,S',P']$ costs $1$ and if at least one of $u$ or $v$ is not allowed to have a neighbor in the other side of the partition in $G[e]$ as per $P'$, then $M_E[e,S',P']$ costs $\infty$. Clearly, every $1$-matched candidate set of $V(e)$ can be considered as a subset of $V(e)$. Thus, number of entries in table $M_E$ is bounded by $n^{O(1)}$ and we can calculate them as per above assignment in time $n^{O(1)}$.

\begin{definition}{\bf (S-compatible set of a bag)}
For $S\subseteq \sigma (t)$ and $\bar{S}= \sigma(t)\setminus S$, a set  $A_s \subseteq \beta (t)$ is called an S-Compatible set of $\beta (t)$ if
\begin{itemize}
    \item $|A_s| \leq k$,
    \item  $A_s \cap \sigma(t) \in \{S, \bar{S}\}$,
    \item $A_s$ is non-empty proper subset of $\beta(t)$ i.e. $A_s\not = \emptyset$ and $A_s\not = \beta(t)$.
\end{itemize}
\end{definition}
\paragraph{}
For an $S$-compatible set $A_s$ of $\beta(t)$, we define $Z(t,A_s) = \{c\ |\ c\in Z(t)\land((A_s\cap \sigma(c)) \neq \emptyset) \land (( \sigma(c)\setminus A_s)\neq \emptyset) \}$ and call it set of \textit{broken children} of $t$ with respect to $A_s$, and  also define $F(t,A_s) = \{e\ |\ e\in E(G_t[\beta(t)])\land ((A_s\cap V(e)) \neq \emptyset) \land (( V(e)\setminus A_s)\neq \emptyset)\}$ and call it set of \textit{broken edges} of $E(G_t[\beta(t)])$ with respect to $A_s$.
Clearly, given $A_s$, we can find $Z(t,A_s)$ and $F(t,A_s)$ in time $n^{O(1)}$.

\begin{definition}{\bf (P-Compatible family)}
For $t,P$ and $A_s$ such that $t\in V(T)$, $P$ is a $d$-matched candidate set of $\sigma(t)$ and $A_s$ is an $S$-compatible set of $\beta(t)$. A family $\calf_{P|A_s}=\{P_c|\ c \in Z(t,A_s)\} \cup \{P_e|\ e \in F(t,A_s)\}$ is called an $A_s$-restricted $P$-compatible family of $t$ if the following holds:

\begin{itemize}
    \item for each $c \in Z(t,A_s)$, is a $d$-matched candidate set of $\sigma(c)$,
     \item for each $e \in F(t,A_s)$, $P_e$ is a $1$-matched candidate set of $V(e)$,
    \item let $P_{z} = \biguplus_{P_v \in \calf_{P|A_s}} P_v$, then $|P_{z} | \leq 2k$,
    \item $\forall v\in V$, $m_{P_z}(v)\leq d$,
    \item $\forall v\in \sigma(t)$, $m_{P_z}(v)\leq m_P(v)$.
\end{itemize}
\end{definition}
\paragraph{}
The intuition behind the $A_s$-restricted $P$-compatible family is as follows.
$P$, $P_c$ and $P_e$ can be considered as a restriction on the possible number of neighbors of vertices of $\sigma(t)$, $\sigma(c)$ and $V(e)$ in other side of a partition respectively in $G_t$, $G_c$ and $G[e]$. $A_s$-restricted $P$-compatible family can be considered as the family of restrictions $P_c$ and $P_e$ which are consistent with $P$ and maintain the property of $d$-matching. 
\paragraph{}
We say that two $A_s$-restricted $P$-compatible families of $t$, $\calf_{P|A_s}=\{P_c|\ c \in Z(t,A_s)\} \cup \{P_e|\ e \in F(t,A_s)\}$  and $\calf'_{P|A_s}=\{P'_c|\ c \in Z(t,A_s)\} \cup \{P'_e|\ e \in F(t,A_s)\}$ are equal if and only if $\forall c \in Z(t,A_s)$, $P_c=P'_c$ and $\forall e \in F(t,A_s)$, $P_e=P'_e$. We say that $\calf_{P|A_s}$ and $\calf'_{P|A_s}$ are distinct if and only if they are not equal.

\begin{proposition}\label{FPAS}
 For an $S$-compatible set $A_s$ of $\beta(t)$, if $|Z(t,A_s)|+|F(t,A_s)|\leq k$, then there are at most $2^{O(k\log k)}$ distinct $A_s$-restricted $P$-compatible families of $t$ and in time $2^{O(k\log k)}n^{O(1)}$ we can list all of them.
\end{proposition}
\begin{proof}
Consider the following procedure to generate $A_s$-restricted $P$-compatible families for $t$. We define $S= \{(c,v,i)\mid c\in Z(t,A_s) \land v\in \sigma(c) \land i\in \{1,2,..,d\}\}$ and $S'=\{(e,v,1)\mid e\in F(t,A_s)\land v\in V(e)\}$. 
Here the triple $(c,v,i)$ indicate $i^{th}$ distinct copy of pair formed by $v$ and $\sigma(c)$ where $v\in \sigma(c)$. Similarly, $(e,v,1)$ indicate the pair formed by $v$ and $V(e)$ where $v\in V(e)$. Given a subset $S^*\subseteq (S\cup S')$, we construct $d$-matched candidate sets of $\sigma(c)$ for every $c\in Z(t,A_s)$ and $1$-matched candidate sets of $V(e)$ for every $e\in F(t,{A_s})$ as follows. For every $c\in Z(t,A_s)$, define $P_c= (V,m_{P_c}) \text{ such that: }$
 \begin{multline}\label{PC_generation}
     \forall v\in \sigma(c), m_{P_c}(v) = |\{(c,v,i)|(c,v,i)\in S^*\}|;
    \forall v\in V\setminus \sigma(c), m_{P_c}(v) =0.
 \end{multline}

That is $m_{P_c}(v)$ equals number of copies of pair $v$ and $\sigma(c)$. If $|P_c|>k$, then we can conclude that $|P_c|$ is not a $d$-matched candidate set and given $S^*$ cannot form $A_s$-restricted $P$-compatible family. Similarly, for every $e\in F(t,{A_s})$, define  $P_e= (V,m_{P_e})\text{ such that: }$
\begin{multline} \label{PE_generation}
\forall v\in V(e), m_{P_e}(v) = |\{(e,v,i)\mid(e,v,i)\in S^* \}|;\forall v\in V\setminus V(e), m_{P_e}(v) =0.
 \end{multline}
 \paragraph{}

 Once we constructed all the $P_c$ and $P_e$ and verified that they are valid $d$-matched candidate set for their respective $\sigma(c)$ and $V(e)$, we can directly check if $\{P_c|\ c \in Z(t,A_s)\} \cup \{P_e|\ e \in F(t,A_s)\}$ satisfies all the conditions of $A_s$-restricted $P$-compatible family or not. Recall that for any $A_s$-restricted $P$-compatible family $F_{P|A_s}$ and corresponding $P_{z} = \biguplus_{P_v \in F_{P|A_s}} P_v$, it should hold that $|P_z|\leq 2k$. This restriction on size allows us to generate all possible $A_s$-restricted $P$-compatible families of $t$ by simply considering all the subsets of $S\cup S'$ of size at most $2k$. If  $|Z(t,A_s)|+|F(t,A_s)|\leq k$, then $|S\cup S'| \leq dk^2 \leq k^3$ as $d$ is at most $k$ and every adhesion $\sigma(c)$ is at most $k$ as well. And the number of subsets of $S\cup S'$ of size at most $2k$ can be bounded by $\sum_{i=0}^{2k} \binom{k^3}{i}$ which is bounded by $2^{O(k \log k)}$, and these subsets can be listed by standard methods which will take time $2^{O(k \log k)}n^{O(1)}$. Thus, there can be at most $2^{O(k \log k)}$ distinct $A_s$-restricted $P$-compatible families of $t$, and they can be listed in time $2^{O(k \log k)}n^{O(1)}$.
 \paragraph{}

 We note that this procedure will generate some repetitions of $A_s$-restricted $P$-compatible families, but we use this procedure for simplicity of analysis and the upper bound obtained from this procedure is sufficient for our purpose. The correctness follows from exhaustive search. This finishes the proof.
 
\qed
 \end{proof}
 
\paragraph{}
We now assume that the entries of the table $M$ table is calculated for every $c\in Z(t)$ and the table $M_E$ is constructed as per above discussed assignment for edges in $E(G_t[\beta(t)])$, note that if $t$ is leaf, then it has no children, yet $M_E$ can be calculated for it as discussed above. For an $S$-compatible set $A_s$ of $\beta(t)$ and $A_s$-restricted $P$-compatible family $\calf_{P|A_s}=\{P_c|\ c \in Z(t,A_s)\} \cup \{P_e|\ e \in F(t,A_s)\}$ of $t$, we define cost of $A_s$ and $\calf_{P|A_s}$ for $t$ as follows.

\begin{equation}\label{cs_calculation}
cs(t,A_s,F_{P|A_s})=\sum_{c\in Z(t,A_s)} M[c, A_s \cap \sigma(c), {P_c},1] + \sum_{e\in F(t,A_s)} M_E[e, A_s \cap V(e), {P_e}].
\end{equation}
 
 \paragraph{}
We define minimum cost of an $S$-compatible set $A_s$ of $\beta(t)$ and $d$-matched candidate set $P$ of $\sigma(t)$ as follows:

\begin{equation}\label{mcs_calculation}
   mcs(t,A_s,P) = min\{ cs(t,A_s, \calf_{P|A_s})| \ \calf_{P|A_s}\text{ is $A_s$-re. $P$-com. family of t}\}.
\end{equation}
\newline

\begin{lemma}\label{MCSProperties}
Given that for every $c\in Z(t)$, the entries corresponding to $c$ of $M$ satisfy property $(1)$ and $(2)$, and the table $M_E$ is constructed as per Construction 3.1, $mcs(t,A_s,P)$ satisfies the following properties.
\begin{enumerate}[(a)]
    \item If $mcs(t,A_s,P) \leq k$, then there exists a partition $(A,B)$ of $G_t$, such that:
    \begin{itemize}
        \item $A \cap \beta(t) = A_s$,
         \item $|E_{G_t}(A,B)|\leq $  $mcs(t,A_s,P)$,
        \item $E_{G_t}(A,B)$ forms a d-matching,
        \item $\forall v\in \sigma(t),\ |N_{G[E_{G_t}(A,B)]}(v)| \leq m_P(v)$.
       
    \end{itemize}
    \item For every partition $(A,B)$ of the entire graph $G$ which satisfies the following prerequisites:
     \begin{itemize}
        \item $A \cap \beta(t) = A_s$,
        \item $|E_G(A,B)|\leq k$,
        \item $E_G(A,B)$ forms a d-matching,
        \item $\forall v\in \sigma(t),\ |N_{G[E_{G_t}(A,B)]}(v)| \leq m_P(v)$.
        
    \end{itemize}
    It holds that $|E_{G_t}(A, B)|\geq mcs(t,A_s,P)$.
\end{enumerate}
\end{lemma}
We note that if $mcs(t,A_s,P)> k$, then there doesn't exist a partition $(A,B)$ of $G$ satisfying prerequisites of property $(b)$.

\begin{proof}
Overall the template of the proof is motivated from \cite[Proofs of Claim 9 and Claim 10]{DBLP:journals/talg/CyganKLPPSW21}, however we need {\kdcut} specific arguments and calculations to derive the proof.
For property ($a$), let $F'_{P|A_s}=\{P_c|\ c \in Z(t,A_s)\} \cup \{P_e|\ e \in F(t,A_s)\}$ be an $A_s$-restricted $P$-compatible family for which $mcs(t,A_s,P) = cs(t,A_s,F'_{P|A_s})$. If $cs(t,A_s,F'_{P|A_s})\leq k$, then  $M[c,A_s\cap \sigma(c),P_c,1] \leq k$ for every $c \in Z(t,A_s)$. Further, for every $c \in Z(t)\setminus Z(t,A_s)$, either $A_s\cap \sigma(c)= \sigma(c)$  or $A_s\cap \sigma(c)= \emptyset$. As per Proposition \ref{empty} we have that $M[c,A_s\cap \sigma(c),\emptyset,0] =0$ for every $c \in Z(t)\setminus Z(t,A_s)$. For every $c\in Z(t,A_s)$, let $(A_c,B_c)$ be the partition of $G_c$ corresponding to property ($1$) of $M[c,A_s\cap \sigma(c),P_c,1]$. And for every $c \in Z(t)\setminus Z(t,A_s)$, let $(A_c,B_c)$ be the partition of $G_c$ corresponding to property ($1$) of $M[c,\emptyset,\emptyset,0]$. Let $Z_a = \{c|\ c\in Z(t)\land (A_c \cap \sigma(c) = A_s\cap \sigma(c))\}$ and $Z_b = \{c|\ c\in Z(t) \land (A_c \cap \sigma(c) = \sigma(c)\setminus A_s)\}$. Observe that $B_c \cap \sigma(c) = A_s\cap \sigma(c)$ for every $c\in Z_b$. We define 
\begin{align*}
    A = (\bigcup_{c\in Z_a} A_c)\ \cup\ (\bigcup_{c\in Z_b} B_c)\ \cup A_s; \ \ \  B =  V(G_t)\setminus A.
\end{align*}
It is now remained to prove that partition $(A,B)$ of $G_t$ satisfies property (\textit{a}) for $mcs(t,A_s,P)$. As $A_c \cap \sigma(c) = A_s\cap \sigma(c)$ for every $c\in Z_a$ and $B_c \cap \sigma(c) = A_s\cap \sigma(c)$ for every $c\in Z_b$, further due to  $\sigma(c)$ being the only vertices that  every $G_c$ shares with $\beta(t)$, and $A_s\subseteq \beta(t)$, we get the property that $A\cap \beta(t) = A_s$.
For every $c\in Z_a$, $A\cap V(G_c)= A_c \cup (A\cap \beta(t)\cap \sigma(c))= A_c\cup (A_s\cap \sigma(c)) = A_c$. Similarly, for every $c\in Z_b$, $A\cap V(G_c)= B_c$. Consider the following claim.

\begin{claim}\label{closure}
 $E_{G_t}(A,B) \subseteq (\bigcup_{c\in Z(t,A_s)}E_{G_c}(A_c,B_c) \bigcup F(t,{A_s}))$.
\end{claim}
Proof of Claim \ref{closure}.
Assume to the contrary that there exist an edge $e=(u,v)\in E_{G_t}(A,B)$ such that $e\not \in (\bigcup_{c\in Z(t,A_s)}E_{G_c}(A_c,B_c) \bigcup F(t,{A_s}))$. Without loss of generality let $u\in A$ and $v\in B$. If $e\in E(G_t[\beta(t)])$, then $u\in A_s$ and $v\in \beta(t)\setminus A_s$ contradicting  $e\not \in F(t,A_s)$. If $e\in E(G_c)$ for a $c\in Z_a$, then due to $A\cap V(G_c)= A_c$ we have $u\in A_c$ and $v\in B_c$ contradicting $e\not \in E_{G_c}(A_c,B_c)$, further $c\not \in Z(t)\setminus Z(t,A_s)$, as for every $c\in Z(t)\setminus Z(t,A_s)$ the corresponding partition $(A_c,B_c)$ is trivial.  If $e\in G_c$ for a $c\in Z_b$, then similar arguments as above holds. This finishes the proof of Claim \ref{closure}.
\paragraph{}

For every edge $e=\{u,v\}\in F(t,A_s)$, $A_s \cap V(e) \in \{u,v\}$, and hence $M_E[e, A_s \cap V(e), {P_e}]\geq 1$ as per the Construction 3.1. This implies that $ |F(t,A_s)|$ is at most $\sum_{e \in F(t,A_s)} M_E[e, A_s \cap V(e), {P_e}]$, and $|\bigcup_{c\in Z(t,A_s)}E_{G_c}(A_c,B_c)|$  is at most $\sum_{c\in Z(t,A_s)}  M[c, A_s \cap \sigma(c), {P_c},1]$. By Claim \ref{closure} we can conclude that $|E_{G_t}(A,B)|$ is at most $cs(t,A_s,F'_{P|A_s})$ which is same as $mcs(t,A_s,P)$.
\paragraph{}

To show that $E_{G_t}(A,B)$ forms a $d$-matching it is sufficient to show that for every vertex $v\in V(G_t)$, $|N_{G[E_{G_t}(A,B)]}(v)|\leq d$. 
Since $cs(t,A_s,F'_{P|A_s})\leq k$, and for every $e \in F(t,A_s)$, $A_s \cap V(e) \in \{u,v\}$, it should hold that $M_E[e, A_s \cap V(e), {P_e}]= 1$ which is as per Construction 3.1 only possible when $m_{P_e}(u) = m_{P_e}(v) = 1$. This ensures that for every vertex $v\in \beta(t)$ and for every edge $e\in F(t,A_s)$, $|N_{G[E_{G_t}(A,B) \cap e]}(v)| \leq m_{P_e}(v)$. Further, for every $c\in Z(t,A_s)$, partition $(A_c,B_c)$ of $G_c$ satisfies property $(1)$ of $M[c,A_s\cap \sigma(c),P_c,1]$, and $\forall v\in \sigma(c)$, $ \ |N_{G[E_{G_c}(A_c,B_c)]}(v)| \leq m_{P_c}(v)$. Recalling $\beta(t) \cap V(G_c) = \sigma(c)$, it holds that for every $v\in \beta(t)$ and every $c\in Z(t,A_s)$, $|N_{G[E_{G_c}(A_c,B_c)]}(v)| \leq m_{P_c}(v)$. Let $P'_{z} = \biguplus_{P_v \in F'_{P|A_s} }P_v$, using Claim \ref{closure} we can conclude that $\forall v \in \beta(t)$, $|N_{G[E_{G_t}(A,B)]}(v)| \leq m_{P'_z}(v)$. As $F'_{P|A_s}$ is an $A_s$-restricted $P$-compatible family, it hods that $\forall v\in V$, $m_{P'_z}(v)\leq d$. This conclude that $\forall v \in \beta(t)$, $|N_{G[E_{G_t}(A,B)]}(v)| \leq d$.
Further, by assumption $E_{G_c}(A_c,B_c) $ forms a $d$-matching for every $c\in Z(t)$. Thus, for any vertex $v \in V(G_c)$, $|N_{G[E_{G_c}(A_c,B_c)}|\leq d$. If $v\not \in \sigma(c)$, then no edge in $E_{G_t}(A,B)\setminus E_{G_c}(A_c,B_c)$ incidents on $v$. Hence, $\forall v \in(V(G_t)\setminus \beta(t))$, $|N_{G[E_{G_t}(A,B)}(v)|\leq d$. This concludes that $E_{G_t}(A,B)$ is a $d$-matching. As $F'_{P|A_s}$ is an $A_s$-restricted $P$-compatible family, it holds that $\forall v\in \sigma(t)$, $m_{P'_z}(v)\leq m_P(v)$. As $\sigma(t)\subseteq \beta(t)$, it holds that  $\forall v\in \sigma(t),\ |N_{G[E_{G_t}(A,B)]}(v)| \leq m_P(v)$. This finishes the proof for property $(a)$.

\paragraph{}
For property $(b)$, let $(A,B)$ be a partition of $G$ which satisfies all the prerequisites of property $(b)$ for given $t$, $A_s$ and $P$. Consider the following assignments.
\begin{align*}
    \forall c\in Z(t,A_s), P_c =(V, m_{P_c}) \text{ such that } \forall v \in  \sigma(c),\  m_{P_c}(v) = |N_{G[E_{G_c}(A,B)]}(v)|,\\
   \text{ and } \forall v \in  V\setminus \sigma(c),\  m_{P_c}(v) = 0;
\end{align*}
and
\begin{align*}
     \forall e\in F(t,A_s), P_e = (V, m_{P_e}) \text{ such that } \forall v \in  V(e),\  m_{P_e}(v) = 1,\\
   \text{ and } \forall v \in  V\setminus V(e),\  m_{P_e}(v) = 0.
\end{align*}

\paragraph{}
 As $E_G(A,B)$ is a $d$-matching, $|N_{G[E_{G_c}(A,B)]}(v)| \leq d$ for every $ v\in \sigma(c)$. Recalling the definition of $G_c$ which excludes every edge with both the endpoints in $\sigma(c)$. Thus, every edge in $E_{G_c}(A,B)$ can have at most one endpoint in $\sigma(c)$ and knowing that $|E_{G_c}(A,B)|\leq k$ we can conclude that $|P_c|\leq k$ and $P_c$ is a $d$-matched candidate set for $\sigma(c)$ for every $c\in Z(t,A_s)$. Further, for every $e\in F(t,A_s)$, $P_e$ is a $1$-matched candidate set for $V(e)$. 
\paragraph{}

We now show that $F_{P|A_s} = \{P_c|\ c\in Z(t,A_s)\} \cup \{P_e|\ e\in F(t,A_s)\}$ is an $A_s$-restricted $P$-compatible family. Let $P_{z} = \biguplus_{P_v \in F_{P|A_s}} P_v$. Observe that all $G_c$ and $G_t[\beta(t)]$ are pairwise edge disjoint 
subgraphs of $G$ because of the exclusion of edges with both endpoints in $\sigma(c)$ in $G_c$. Combining the fact that $\sum_{v\in V(G)} |N_{G[E(A,B)]}(v)| = 2|E(A,B)|$, $|E(A,B)|$ is at most $k$ and $E(A,B)$ is a $d$-matching, we can assert that $|P_z|\leq 2k$ and $m_{P_z}(v)\leq d$ for every $ v \in V$. Further, as $(A,B)$ satisfies prerequisites of property $(b)$, it is known that $|N_{G[E_{G_t}(A,B)]}(v)| \leq m_P(v)$ for every $v\in \sigma(t)$. Again, the fact that all $G_c$ and $G_t[\beta(t)]$ are pairwise edge disjoint subgraphs of $G_t$ helps us assert that $m_{P_z}(v) \leq m_P(v)$ for every  $ v\in \sigma(t)$. This concludes that $F_{P|A_S}$ is an $A_s$-restricted $P$-compatible family for $t$.
\paragraph{}

We move on to show that $(A,B)$ satisfies prerequisites of property $(2)$ for $M[c,A_s\cap \sigma(c),P_c,1]$ for every $c\in Z(t,A_s)$. As $A_s\cap \sigma(c) = A\cap\beta(t)\cap\sigma(c)$, and as $\sigma(c) \subseteq \beta(t)$ we have $A\cap \sigma(c) = A_s\cap \sigma(c)$. Satisfaction of rest of the prerequisites can be directly checked using the fact that $(A,B)$ satisfies prerequisites of property $(b)$ for $mcs(t,A_s,P)$ and recalling the assignment of $P_c$, further it is straightforward to see that both $A\cap V(G_c)$ and $B\cap V(G_c)$ are non-empty for every $c\in Z(t,A_s)$ because $A_s\cap \sigma(c)$ and $\sigma(c)\setminus A_s$ are both non-empty as per the definition of  $Z(t,A_s)$ and $A\cap \sigma (c) = A_s\cap \sigma(c)$ for every $c\in Z(t,A_s)$.
\paragraph{}

For an edge $e=(u,v)\in F(t,A_s)$, we can directly check that $M_E[t,A_s\cap V(e),P_e]=1$ as per the Construction 3.1, as in this case $P_e(u) = P_e(v) = 1$ and $(A_s\cap V(e)) = \{u\} $ or $\{v\}$. Thus, recalling that all $G_c$ and $G_t[\beta(t)]$ are all pairwise edge disjoint, we conclude that $|E_{G_t}(A,B)|$ is at least $cs(t,A_s,F_{P|A_s})$ calculated as per equation (\ref{cs_calculation}). Further, $mcs(t,A_s,P) \leq cs(t,A_s,F_{P|A_s})$ due to the minimality over all $A_s$- restricted $P$-compatible families. This finishes the proof.
\qed
\end{proof}

\begin{proposition}\label{broken}
For every $A_s$-restricted $P$-compatible family $\calf_{P|A_s}$ of $t$, \\ $cs(t,A_s,F_{P|A_s}) \geq |Z(t,A_s)|+ |F(t,A_s)|$.
\end{proposition}
\begin{proof}
We use arguments similar to \cite[Proof of Claim 11]{DBLP:journals/talg/CyganKLPPSW21}.
It is sufficient to show that every $c\in Z(t,A_s)$ or $e \in F(t,A_s)$ contributes a positive value to the computation of $cs(t,A_s,F_{P|A_s})$ in Equation (\ref{cs_calculation}). By Construction 3.1, for every $e=(u,v)\in F(t,A_s)$, $M_E[e,A_s\cap V(E),P_e]\geq 1$  as $A_s\cap V(E) =\{u\}$ or $A_s\cap V(E) =\{v\}$.
\paragraph{}

For every $c \in Z(t,A_s)$, observe that $c$ contributes $M[c,A_s\cap \sigma(c),P_c]$ in equation (\ref{cs_calculation}). For every $c \in Z(t,A_s)$ both $A_S\cap \sigma(c)$ and $\sigma(c) \setminus A_s$ are non empty. Recalling proposition \ref{nonempty} we conclude $M[c,A_s\cap \sigma(c),P_c]\geq 1$. This finishes the proof.
\qed \end{proof}

\begin{claim} \label{cscalcbond}
Given that the entries of $M$ being computed for every $c\in Z(t)$, and $M_E$ is constructed as per Construction 3.1. For an $S$-compatible set $A_s$ of $\beta(t)$ and an $A_s$-restricted $P$-compatible family $\calf_{P|A_s}$ of $t$, $cs(t,A_s,F_{P|A_s})$ can be computed in time  $2^{O(k\log k)}n^{O(1)}$. 
\end{claim}
\begin{proof}
As discussed earlier,  computation of $Z(t,A_s)$ and $F(t,A_s)$ will take $n^{O(1)}$ time.
 Number of entries in $M$ corresponding to every $c\in Z(t)$ are bounded by $2^{O(k\log k)}$, further there are $O(1)$ entries in $M_E$ for every $e\in E(G_t[\beta(t)])$. To calculate  $cs(t,A_s,F_{P|A_s})$ as per equation $\ref{cs_calculation}$, we need to retrieve at most $|Z(t,A_s)|+ |F(t,A_s)| \leq n^{O(1)}$ entries for  computation as per equation $\ref{cs_calculation}$. Assuming we do a linear search, the time taken to retrieve required entries is bounded by $2^{O(k\log k)}n^{O(1)}$.
\qed \end{proof}

\begin{lemma}\label{mcs_lemma}
For an $S$-compatible set $A_s$ of $\beta(t)$, in time $2^{O(k\log k)}n^{O(1)}$ we can either decide that $mcs(t,A_s,P) > k$ or calculate $mcs(t,A_s,P)$.
\end{lemma}
\begin{proof}
For given $A_s$, we check if $|Z(t,A_s)|+ |F(t,A_s)| \leq k$, if not, then using the Proposition \ref{broken} we conclude that $mcs(t,A_s,P) > k$ . Else, if $|Z(t,A_s)|+ |F(t,A_s)| \leq k$, then we use proposition $\ref{FPAS}$ to get all the $2^{O(k\log k)}$ distinct $A_s$- restricted $P$-compatible families of $t$ in time  $2^{O(k\log k)}n^{O(1)}$ and calculate  $mcs(t,A_s,P)$ as per equation \ref{mcs_calculation}, where we need to compute $cs(t,A_s,F_{P|A_s})$ for $2^{O(k\log k)}$ distinct $A_s$- restricted $P$-compatible families, each of which can be computed using claim $\ref{cscalcbond}$ in time  $2^{O(k\log k)}n^{O(1)}$. Thus, we conclude that  computation of $mcs(t,A_s,P)$ would take time $2^{O(k\log k)}n^{O(1)}$.
\qed \end{proof}
\paragraph{}

We now move on to give an assignment to $M[t,S,P,1]$. Consider the following.
 \begin{multline}\label{MINC}
    \textsc{MIN}_c = \min \{\min \{ M[c,\emptyset,P_c,1]\ |\ P_c \text{ is a $d$- mat. can. set of } \sigma(c) \text{ such that }\\ \forall v\in \sigma(t), m_{P_c}(v)\leq m_P(v)\}\ \ |\  c\in Z(t)\}.
\end{multline}
\ If $t$ is a leaf node and that $Z(t)$ is empty, then we set $\textsc{MIN}_c=\infty$.
\begin{multline}\label{MINBT}
\text{MIN}_{\beta(t)} = \min\{ mcs(t,A_s,P)\ |\  A_s  \text{ is an }S\text{ -compatible set of } \beta(t) \}.
\end{multline}

Consider the following assignment of $M[t,S,P,1]$.

\begin{enumerate}
    \item Case: $S=\emptyset$ or $S=\sigma(t)$.

    \begin{equation}\label{mtsp_empty}
        M[t,S,P,1] = min\{ \text{MIN}_c, \text{MIN}_{\beta(t)}\}.
    \end{equation}
    
    \item Case: $S\not =\emptyset$ and $S\not =\sigma(t)$.
    \begin{equation}\label{mtsp_nonempty}
         M[t,S,P,1] = \text{MIN}_{\beta(t)}.
    \end{equation}
  We note that in equations (\ref{mtsp_empty}) and (\ref{mtsp_nonempty}) if the assignment of $ M[t,S,P,1]$ is $>k$, then we set $M[t,S,P,1]=\infty$.
\end{enumerate}

We now argue the correctness of the assignment of $ M[t,S,P,1]$.

\begin{lemma}\label{M[t,S,P]_correctness}
Given that the entries of $M$ correspond to every every $c\in Z(t)$ satisfy property $(1)$ and $(2)$, and  $mcs(t,A_s,P)$ satisfies property $(a)$ and $(b)$ for every $S$-compatible set $A_s$ of $\beta(t)$, assignment of $M[t,S,P,1]$ as per equation (\ref{mtsp_empty}) and (\ref{mtsp_nonempty}) satisfies property $(1)$ and $(2)$.
\end{lemma}
\begin{proof}

Overall the template of the proof is motivated from \cite[Proofs of Claim 9 and Claim 10]{DBLP:journals/talg/CyganKLPPSW21}, however we need {\kdcut} specific arguments and calculations to derive the proof.
We will divide the proof in the following cases.
For property $(1)$, consider the following two cases.\\
\textbf{Case 1a:}  If $M[t,S,P,1] = \text{MIN}_{\beta(t)}$.\\
$M[t,S,P,1]\leq k$, then let $A_s$ be the $S$-compatible set such that $M[t,S,P,1]= mcs(t,A_s,P)$. Let $(A,B)$ be the partition of $G_t$ corresponding to property $(a)$ for $mcs(t,A_s,P)$. As $A\cap \beta(t) = A_s$, this implies $A \cap \sigma(t) = A_s \cap \sigma(t)$, and due to $A_s$ being an $S$-compatible set, we can conclude $A \cap \sigma(t) \in \{S,\bar{S}\}$. Rest of the statements of property $(1)$ are directly satisfied by partition $(A,B)$ if property $(a)$ is satisfied by $(A,B)$. Further as $A_s$ is non-empty proper subset of $\beta(t)$, both $A$ and $B$ are non-empty. \\
\textbf{Case 1b:} $M[t,S,P,1] = \text{MIN}_{c}$.
This is only possible when $t$ is not a leaf (otherwise $\text{MIN}_c$ would be $\infty$). In this case there exist a $c\in Z(t)$ and $P_{c}$ such that $\forall v\in \sigma(t), m_{P_{c}}(v)\leq m_P(v)$ and $M[t,S,P,1]= M[c,\emptyset,P_{c},1]\leq k$ . Let $(A_{c},B_{c})$ be the partition of $G_{c}$ corresponding to property $(1)$ of $M[c,\emptyset,P_{c^*},1]$, then either $A_{c}\cap \sigma(c) = \emptyset$ or $A_{c}\cap \sigma(c) = \sigma(c)$. In the first case we can verify that $(A_{c},B_{c}\cup (V(G_t)\setminus V(G_{c}))$ is a partition of $G_t$ which satisfies all the points of property $(1)$ for $M[t,S,P,1]$ as $(A_{c},B_{c})$ satisfies property $(1)$ for $M[c,\emptyset,P_{c},1]$ and every edge in $E_{G_t}(A_{c},B_{c}\cup (V(G_t)\setminus V(G_{c}))$ belong to $E_{G_{c}}(A_{c},B_{c})$. In the second case 
$(A_{c}\cup (V(G_t)\setminus V(G_{c}),B_{c})$ is a partition of $G_t$ which satisfies all the points of property $(1)$ for $M[t,S,P,1]$ as $(A_{c},B_{c})$ satisfies property $(1)$ for $M[c,\emptyset,P_{c},1]$, and every edge in $E_{G_t}(A_{c}\cup (V(G_t)\setminus V(G_{c}),B_{c})$ in the second case belong to $E_{G_{c}}(A_{c},B_{c})$. Further, in the first case the partition $(A_{c},B_{c}\cup (V(G_t)\setminus V(G_{c}))$, and in the second case the partition $(A_{c}\cup (V(G_t)\setminus V(G_{c}),B_{c})$  are non trivial as $(A_{c},B_{c})$ is non trivial. This concludes property $(1)$.
\paragraph{}

For property $(2)$, let $(A,B)$ be the partition of $G$ which satisfies all the prerequisites of $(2)$ for given $t$, $S$ and $P$.\\
\textbf{Case 2a:} both $A\cap \beta(t)$ and $B\cap \beta(t)$ are non-empty. We now prove that $E_{G_t}(A,B) \geq M[t,S,P,1]$. By $(k,k)$-edge unbreakability of $\beta(t)$, we have at least one of $A\cap \beta(t)\leq k$ or $B\cap \beta(t)\leq k$.
Let $A_s = A\cap \beta(t)$ and $B_s = B\cap \beta(t)$. In this case, $A_s$ (otherwise $B_s$) is a non-empty proper subset of $\beta(t)$ and $|A_s|$ (otherwise $|B_s|$) is at most $k$, thus $A_s$ (otherwise $B_s$) is an $S$-compatible set of $\beta(t)$ as $A_s\cap \sigma(t) \in \{S, \bar{S}\}$ (otherwise $B_s\cap \sigma(t) \in \{S, \bar{S}\}$) as $A\cap \sigma(t) \in \{S,\bar{S}\}$ which also implies $B\cap \sigma(t) \in \{\bar{S},S\}$.  we can directly check the satisfaction of the prerequisites of property $(b)$ for $mcs(t,A_s,P)$ by partition $(A,B)$  (otherwise by $(B,A)$). 
Thus, $|E_{G_t}(A,B)| \geq mcs(t,A_s,P)$ (otherwise $\geq mcs(t,B_s,P)$), and due to minimality $M[t,S,P,1]\leq \text{MIN}_{\beta(t)} \leq mcs(t,A_s,P)$ (otherwise $\leq mcs(t,B_s,P)$) hence  $|E_{G_t}(A,B)| \geq M[t,S,P,1]$. \\
\textbf{Case 2b:} $A\cap \beta(t) = \emptyset$ or $A\cap \beta(t) = \beta(t)$. Note that this can only be possible if $t$ is a non leaf node since $A\cap V(G_t)$ and $B\cap V(G_t)$ are non empty. 
Consider the following.
\begin{align*}
    \forall c\in Z(t), P_c =(V, m_{P_c}) \text{ such that } \forall v \in  \sigma(c),\  m_{P_c}(v) = |N_{G[E_{G_c}(A,B)]}(v)|;\\
   \text{ and } \forall v \in  V\setminus \sigma(c),\  m_{P_c}(v) = 0;
\end{align*}

As $E_G(A,B)$ is a $d$-matching, hence for every $c\in Z(t)$, and for every $ v\in \sigma(c)$, $|N_{G[E_{G_c}(A,B)]}(v)| \leq d$. Recalling that $G_c$ has no edge with both the endpoints in $\sigma(c)$, every edge in $E_{G_c}(A,B)$ has at most one endpoint in $\sigma(c)$ and since $|E_{G_c}(A,B)|\leq k$, we can conclude that $|P_c|\leq k$ and $P_c$ is a $d$-matched candidate set for $\sigma(c)$ .
Further,  $\forall v \in  \sigma(t),\  m_{P_c}(v) \leq m_{P}(v)$, this is because $(A,B)$ satisfies prerequisites of property $(2)$ for $M[t,S,P,1]$, and thus $\forall v\in \sigma(t), |N_{G[E_{G_t}(A,B)]}(v)|\leq m_{P}(v)$, and since $G_c$ is a subgraph of $G_t$ we conclude $\forall v \in  \sigma(t),\  m_{P_c}(v) \leq m_{P}(v)$.
Now we can verify that $(A,B)$ satisfies all the prerequisites of property $(2)$ for $M[c,\emptyset,P_c,n_e]$ for every $c\in Z(t)$. Further there must be at least one $c\in Z(t)$ such that both $A\cap V(G_c)$ and $B\cap V(G_c)$ are non empty. This is guaranteed because $(A,B)$ is a non trivial partition, and $\beta(t)$ completely belong to either $A$ or $B$. Let $c^*$ be that child of $t$ for which both $A\cap V(G_{c^*})$ and $B\cap V(G_{c^*})$ are non empty, then $E_{G_{c^*}}(A,B)\geq M[c^*,\emptyset,P_{c^*},1]$. Further, $M[c^*,\emptyset,P_{c^*},1]$ must be considered for  computation of $\text{MIN}_c$ in equation $(\ref{MINC})$. Thus, $\text{MIN}_c\leq M[c^*,\emptyset,P_{c^*},1]$. Further, as $E_{G_{c^*}}(A,B)\subseteq E_{G_t}(A,B)$. We have that $|E_{G_t}(A,B)|\geq \text{MIN}_c \geq M[t,S,P,1]$.
\qed \end{proof}

\paragraph{}
 computation of $\text{MIN}_c$ is straightforward and requires to iterate over entries of $M$ corresponding to every $c\in Z(t)$. As the number of entries in $M$ for each $c$ are bounded by $2^{O(k\log k)}$ and $|Z(t)|$ can be at most $O(n)$, we can compute $\text{MIN}_c$ in $2^{O(k\log k)}n^{O(1)}$.
\paragraph{}

To calculate $\textsc{MIN}_{\beta(t)}$ a simple brute force approach of guessing all the $S$-compatible sets of $\beta(t)$ will not work, as it will breach the running time budget that we have. However, as it is required to compute $\text{MIN}_{\beta(t)}$ only if $\text{MIN}_{\beta(t)}\leq k$, and thus the search space can be restricted. We now move on to compute $\text{MIN}_{\beta(t)}$, to compute $\text{MIN}_{\beta(t)}$ we follow the framework given in  \cite{DBLP:journals/talg/CyganKLPPSW21} which requires construction of an auxiliary graph and use of color coding technique (Lemma \ref{color}), which is as follows.
\paragraph{}
Let us assume that $\text{MIN}_{\beta(t)} \leq k$, and let  $A^*_s$ be the $S$-compatible set such that $\text{MIN}_{\beta(t)} = mcs(t,A^*_s,P)\leq k$.
In such a scenario due to proposition \ref{broken} we can assume that $|Z(t,A_s)| + |F(t,A_s)| \leq k$. Let $B^* = (\cup_{c\in Z(t,A^*_s)} (\sigma(c)\setminus A^*_s)) \bigcup (\cup_{e\in  F(t,A^*_s)} (V(e)\setminus A^*_s))$. Due to $|Z(t,A^*_s)| + |F(t,A^*_s)| \leq k$ and $|\sigma(c)|\leq k$,  we can observe that $|B^*| \leq k^2$ . By lemma \ref{color}, for the universe $\beta(t)$ and integers $k,k^2+k$, we obtain a family $\calf$ of subsets of $\beta(t)$ such that there exist a set $A_g \in \calf$ such that $A_g \supseteq A^*_s$ and $A_g \cap (B^*\cup (\sigma(t)\setminus A^*_s)) = \emptyset$. We call such set $A_g$ a \textit{good set}. Further, the size of $\calf$ is bounded by $2^{O(k\log k)}\log n$.
\paragraph{}

We construct an auxiliary graph $H$ on vertex set $\beta(t)$ and add an edge $(u,v)\in E[H]$ if and only if one of the following holds,
\begin{enumerate}
    \item $u,v\in \sigma(t)$;
    \item there exist a $c \in Z(t)$ such that  $u,v\in \sigma(c)$;
    \item $(u,v)\in E(G_t[\beta(t)])$.
\end{enumerate}
Observe that $\sigma(t)$ forms a clique in $H$, similarly every $\sigma(c)$ forms a clique in $H$ and $G_t[\beta(t)]$ is a subgraph of $H$. For $X\subseteq \beta(t)$, we call a connected component $C_s$ of $H[X]$ an $S$\textit{-compatible component} if $V(C_s)$ is an $S$-compatible set of $\beta(t)$.

\begin{proposition}\label{gcS}
 If $A_g$ is a good set, then there exist an $S$-compatible component $C_s$ in the subgraph $H[A_g]$ such that $mcs(t,A^*_s,P)= mcs(t,V(C_s),P)$.

\end{proposition}
\begin{proof} We first make a claim similar to \cite[Claim 12]{DBLP:journals/talg/CyganKLPPSW21} which will be used for the proof of the above proposition, the claim states that no component $C$ in $H[A_g]$ both contain a vertex of $A^*_s$ and a vertex outside $A^*_s$ at the same time.  Assume to the contrary that there exist a connected component $C\in H[A_g]$ such that $V(C)\cap A^*_s \not = \emptyset$ and $V(C)\setminus A^*_s \neq \emptyset$. In such a case there must be an edge $e=\{u,v\}$ in $H[A_g]$ such that $u\in A^*_s$ and $v\not \in A^*_s$. Recalling the construction of $H$, edge $\{u,v\}$ may be present due to (i) both $u$ and $v$ belong to $\sigma(t)$, in which case it contradicts that $A_g$ is a good set because it violates the condition that $A_g\cap (\sigma(t)\setminus A^*_s)= \emptyset$. (ii) Both $u$ and $v$ belong to $\sigma(c)$ for a same $c\in Z(t)$, this also contradicts that $A_g$ is a good set as it violates the condition that $A_g\cap B^* = \emptyset$. (iii) $\{u,v\}$ is an edge in $E(G_t[\beta(t)])$, this implies $\{u,v\}\in F(t,A^*_s)$ which contradicts that $A_g$ is a good set as it also violates the condition that $A_g\cap B^* = \emptyset$.
\paragraph{}
  Now for the proof of the above proposition, let $\{C_1,C_2,...,C_l\}$ be the set of all the connected components in $H[A_g]$ such that $V(C_i)\subseteq A^*_s$ for every $1\leq i \leq l$, then $ A^*_s= \bigcup_{1\leq i\leq l} V(C_i)$ as $A_g\supseteq A^*_s$ . Let us define $Z_i = \{c\ |\ c\in Z(t,A^*_s) \land (V(C_i)\cap \sigma(c)\neq \emptyset)\}$ and $F_i = \{e\ |\ e\in F(t,{A^*_s}) \land (V(C_i)\cap V(e) \neq \emptyset)\}$. A crucial fact is that the sets $Z_i$ are pairwise disjoint with each other due to the construction of $H$ where $\sigma(c)$ forms a clique in $H$ for every $c$. Similarly, sets $F_i$ are also pairwise disjoint. This implies that $A^*_s\cap \sigma(c) = V(C_i)\cap \sigma(c)$ for every $c\in Z_i$ and $A^*_s\cap V(e) = V(C_i)\cap V(e)$ for every $e\in F_i$. Let $\calf_{P|A^*_s} = \{P_c| \ c \in Z(t,{A^*_s}) \} \cup \{P_e| \ e \in F(t,A^*_s)\}$ be the family for which $mcs(t,A^*_s,P)= cs(t,A^*_s,F_{P|A^*_s})$, then

\begin{multline*}\label{cs_components}
     cs(t,A^*_s,F_{P|A_s})=\sum_{1\leq i\leq l}\sum_{c\in Z_i} M[c, V(C_i) \cap \sigma(c), {P_c}]\\ + \sum_{1\leq i\leq l} \sum_{e\in F_i} M_E[e, V(C_i) \cap V(e), {P_e}].
\end{multline*}

\paragraph{}

We now define the families $\calf_{P|V(C_i)} = \{P_c| \ P_c \in \calf_{P|A^*_s} \land c\in Z_i\}\cup \{P_e| \ P_e \in {\calf}_{P|A^*_s} \land e\in F_i\}$. We now consider two cases, in the first case we assume that both $S$ and $\bar{S}$ are non empty. Observe that there exist a component $C_s$ such that  $V(C_s)\cap \sigma(t) = (A^*_s\cap \sigma(t))$ which is essentially $S$ or $\bar{S}$. This is guaranteed because $A_g\cap \sigma(t) = A^*_s\cap \sigma(t)$ as $A_g$ is disjoint from $\sigma(t)\setminus A^*_s$ and  $\sigma(t)$ forms a clique in $H$, thus $A_g\cap \sigma(t)$ should be connected in $H[A_g]$. Further, $V(C_s)\subseteq A^*_s$ as we already discussed that for a component $C$ either $V(C)\subseteq A^*_s$ or $V(C)\cap A^*_s = \emptyset$ if $A_g$ is a good set, further $V(C_s)$ is non-empty and proper subset of $\beta(t)$ as both $S$ and $\bar{S}$ are non-empty. Further, as $V(C_s)\subseteq A^*_s$ we can observe that $|V(C_s)|\leq k$. Thus, $C_s$ qualifies to be an $S$-compatible component. Now we have, 
\begin{multline*}
    cs(t,A^*_s,\calf_{P|A_s})=cs(t,V(C_s),{\calf}_{P|V(C_s)}) +\\ \sum_{1\leq i\leq l \land i\neq s}\left(\sum_{c\in Z_i} M[c, V(C_i) \cap \sigma(c), {P_c},1] + \sum_{e\in E_i} M_E[e, V(C_i) \cap V(e), {P_e}]\right).
\end{multline*}
\normalsize
Thus, we get $cs(t,V(C_s),\calf_{P|V(C_s)})\leq cs(t,A^*_s,\calf_{P|A^*_s})$. Due to minimality,\\ $mcs(t,V(C_s),P)$\ $\leq cs(t,V(C_s),\calf_{P|V(C_s)})$. And due to the minimality of\\ $ mcs(t,A^*_s,P)$ among all $S$-compatible sets, we conclude that $mcs(t,V(C_s),P)$ = $mcs(t,A^*_s,P)$.

\paragraph{}
In the second case, if either $S$ or $\bar{S}$ is empty, then for every $C_i$ such that $V(C_i)\subseteq A^*_s$, we have that $V(C_i)\cap \sigma(t) \in \{\sigma(t),\emptyset\}$, which is essentially $S$ or $\bar{S}$. Further, $V(C_i)\neq \emptyset$, $V(C_i)\subseteq A^*_s$ and $A^*_s\subsetneq \beta(t)$, $V(C_i)$ is a non-empty proper subset of $\beta(t)$. Thus, every $C_i$ is an $S$-compatible set. And we have

\begin{align*}
    cs(t,A^*_s,\calf_{P|A^*_s})=\sum_{1\leq i \leq l }  cs(t,V(C_i),\calf_{P|V(C_i)}).
\end{align*}

Consider any component $C_i$, if $V(C_i)= A^*_s$, then we are done, else if\\ $V(C_i)\subsetneq A^*_s$, then we have $cs(t,V(C_i),\calf_{P|V(C_i)})\leq cs(t,A^*_s,\calf_{P|A^*_s})$. This implies\\ $mcs(t,V(C_i),P)\leq mcs(t,A^*_s,P)$. But, due to minimality of $cs(t,A^*_s,P)$, we have $mcs(t,V(C_i),P)= mcs(t,A^*_s,P)$. This finishes the proof.
\qed \end{proof}

\paragraph{}
\normalsize

\normalsize
 Proposition $\ref{gcS}$ allow us to efficiently calculate $MIN_{\beta(t)}$. 
We need to iterate over every $A_g\in \calf$ and for each $S$-compatible component $C_s$ in $H[A_g]$ (if such $C_s$ exist in $H[A_g]$) we need to use Lemma $\ref{mcs_lemma}$ to either calculate $mcs(t,V(C_s),P)$ or decide if $mcs(t,V(C_s),P)>k$. If $mcs(t,V(C_s),P)>k$, then we assume it to be $\infty$. We take the minimum value $mcs(t,V(C_s),P)$ encountered among all the $S$-compatible component $C_s$ in $H[A_g]$ over all the choices $A_g\in \calf$ and assign it to $\text{MIN}_{\beta(t)}$. Correctness of this procedure comes due to the minimality of $mcs(t,A^*_s,P)$ among all the $S$-compatible sets of $\beta(t)$ and due to Proposition $\ref{gcS}$.
If we don't encounter any $S$-compatible component during this process, then we can conclude that the assumption $\text{MIN}_{\beta(t)}\leq k$ doesn't hold and we set $\text{MIN}_{\beta(t)}=\infty$.
\paragraph{}

As the size of $\calf{}$ is bounded by $2^{O(k\log k)}\log n$ and we can obtain it using Lemma \ref{color} in time $2^{O(k\log k)}n\log n$. And for every $A_g\in \calf{}$, $H[A_g]$ can contain at most $n$ $S$-compatible components and we can find all of them in time $n^{O(1)}$ by using standard graph traversal methods. Thus, we need to use Lemma $\ref{mcs_lemma}$ for at most $2^{O(k\log k)}n^{O(1)}$ $S$-compatible components(sets), and each use takes time $2^{O(k\log k)}n^{O(1)}$, thus  computation of $\text{MIN}_{\beta(t)}$ takes time $2^{O(k\log k)}n^{O(1)}$. Recalling that  computation of $\text{MIN}_c$ takes $2^{O(k\log k)}n^{O(1)}$. This conclude that a single entry $M[t,S,P,1]$ can be computed in time $2^{O(k\log k)}n^{O(1)}$. Further we use proposition $\ref{empty}$ and $\ref{infty}$ to set values of $M[t,S,P,0]$. This concludes that a single entry $M[t,S,P,n_e]$ can be calculated in time $2^{O(k\log k)}n^{O(1)}$. 
Recalling Lemma $\ref{table_correctness}$, this suffices to conclude the proof of Theorem $\ref{FPTMC}$.

\section{Conclusion}
We obtained a $2^{O(k\log k)}n^{O(1)}$ time fixed parameter tractable algorithm for {\kdcut}. And we have that the problem of deciding if the input graph admits a matching cut with edge cut of size at most $k$ can not be solved in time $2^{o(k)}n^{O(1)}$  unless ETH fails. It will be an interesting problem to reduce the gap between lower and upper bound of {\matchingcut}.

\paragraph{}
\textbf{Acknowledgements.}
We thank Fahad Panolan for useful discussions, in particular his suggestion of the compact tree decomposition with bounded adhesion and guaranteed unbreakability that we used in this paper.
\bibliographystyle{plainurl}
\bibliography{d-cut}

\begin{thebibliography}{10}

\bibitem{DBLP:conf/cocoa/AravindKK17}
N.~R. Aravind, Subrahmanyam Kalyanasundaram, and Anjeneya~Swami Kare.
\newblock On structural parameterizations of the matching cut problem.
\newblock In Xiaofeng Gao, Hongwei Du, and Meng Han, editors, {\em Combinatorial Optimization and Applications - 11th International Conference, {COCOA} 2017, Shanghai, China, December 16-18, 2017, Proceedings, Part {II}}, volume 10628 of {\em Lecture Notes in Computer Science}, pages 475--482. Springer, 2017.
\newblock \href {https://doi.org/10.1007/978-3-319-71147-8\_34} {\path{doi:10.1007/978-3-319-71147-8\_34}}.

\bibitem{DBLP:journals/jgt/Bonsma09}
Paul~S. Bonsma.
\newblock The complexity of the matching-cut problem for planar graphs and other graph classes.
\newblock {\em J. Graph Theory}, 62(2):109--126, 2009.
\newblock \href {https://doi.org/10.1002/jgt.20390} {\path{doi:10.1002/jgt.20390}}.

\bibitem{DBLP:conf/focs/ChitnisCHPP12}
Rajesh~Hemant Chitnis, Marek Cygan, MohammadTaghi Hajiaghayi, Marcin Pilipczuk, and Michal Pilipczuk.
\newblock Designing {FPT} algorithms for cut problems using randomized contractions.
\newblock In {\em 53rd Annual {IEEE} Symposium on Foundations of Computer Science, {FOCS} 2012, New Brunswick, NJ, USA, October 20-23, 2012}, pages 460--469. {IEEE} Computer Society, 2012.
\newblock \href {https://doi.org/10.1109/FOCS.2012.29} {\path{doi:10.1109/FOCS.2012.29}}.

\bibitem{Chvatal84}
Vasek Chv{\'{a}}tal.
\newblock Recognizing decomposable graphs.
\newblock {\em Journal of Graph Theory}, 8(1):51--53, 1984.
\newblock \href {https://doi.org/10.1002/jgt.3190080106} {\path{doi:10.1002/jgt.3190080106}}.

\bibitem{DBLP:journals/iandc/Courcelle90}
Bruno Courcelle.
\newblock The monadic second-order logic of graphs. i. recognizable sets of finite graphs.
\newblock {\em Inf. Comput.}, 85(1):12--75, 1990.
\newblock \href {https://doi.org/10.1016/0890-5401(90)90043-H} {\path{doi:10.1016/0890-5401(90)90043-H}}.

\bibitem{DBLP:books/sp/CyganFKLMPPS15}
Marek Cygan, Fedor~V. Fomin, Lukasz Kowalik, Daniel Lokshtanov, D{\'{a}}niel Marx, Marcin Pilipczuk, Michal Pilipczuk, and Saket Saurabh.
\newblock {\em Parameterized Algorithms}.
\newblock Springer, 2015.
\newblock \href {https://doi.org/10.1007/978-3-319-21275-3} {\path{doi:10.1007/978-3-319-21275-3}}.

\bibitem{DBLP:journals/talg/CyganKLPPSW21}
Marek Cygan, Pawel Komosa, Daniel Lokshtanov, Marcin Pilipczuk, Michal Pilipczuk, Saket Saurabh, and Magnus Wahlstr{\"{o}}m.
\newblock Randomized contractions meet lean decompositions.
\newblock {\em {ACM} Trans. Algorithms}, 17(1):6:1--6:30, 2021.
\newblock \href {https://doi.org/10.1145/3426738} {\path{doi:10.1145/3426738}}.

\bibitem{DBLP:journals/algorithmica/GomesS21}
Guilherme de~C.~M.~Gomes and Ignasi Sau.
\newblock Finding cuts of bounded degree: Complexity, {FPT} and exact algorithms, and kernelization.
\newblock {\em Algorithmica}, 83(6):1677--1706, 2021.
\newblock URL: \url{https://doi.org/10.1007/s00453-021-00798-8}, \href {https://doi.org/10.1007/S00453-021-00798-8} {\path{doi:10.1007/S00453-021-00798-8}}.

\bibitem{DBLP:books/daglib/0030488}
Reinhard Diestel.
\newblock {\em Graph Theory, 4th Edition}, volume 173 of {\em Graduate texts in mathematics}.
\newblock Springer, 2012.

\bibitem{DBLP:series/txcs/DowneyF13}
Rodney~G. Downey and Michael~R. Fellows.
\newblock {\em Fundamentals of Parameterized Complexity}.
\newblock Texts in Computer Science. Springer, 2013.
\newblock \href {https://doi.org/10.1007/978-1-4471-5559-1} {\path{doi:10.1007/978-1-4471-5559-1}}.

\bibitem{graham1970primitive}
Ron~L Graham.
\newblock On primitive graphs and optimal vertex assignments.
\newblock {\em Ann. New York Acad. Sci}, 175:170--186, 1970.

\bibitem{DBLP:conf/cocoon/HsiehLLP19}
Sun{-}Yuan Hsieh, Ho{\`{a}}ng{-}Oanh Le, Van~Bang Le, and Sheng{-}Lung Peng.
\newblock Matching cut in graphs with large minimum degree.
\newblock In Ding{-}Zhu Du, Zhenhua Duan, and Cong Tian, editors, {\em Computing and Combinatorics - 25th International Conference, {COCOON} 2019, Xi'an, China, July 29-31, 2019, Proceedings}, volume 11653 of {\em Lecture Notes in Computer Science}, pages 301--312. Springer, 2019.
\newblock \href {https://doi.org/10.1007/978-3-030-26176-4\_25} {\path{doi:10.1007/978-3-030-26176-4\_25}}.

\bibitem{DBLP:journals/dam/KomusiewiczKL20}
Christian Komusiewicz, Dieter Kratsch, and Van~Bang Le.
\newblock Matching cut: Kernelization, single-exponential time fpt, and exact exponential algorithms.
\newblock {\em Discret. Appl. Math.}, 283:44--58, 2020.
\newblock \href {https://doi.org/10.1016/j.dam.2019.12.010} {\path{doi:10.1016/j.dam.2019.12.010}}.

\bibitem{DBLP:journals/tcs/KratschL16}
Dieter Kratsch and Van~Bang Le.
\newblock Algorithms solving the matching cut problem.
\newblock {\em Theor. Comput. Sci.}, 609:328--335, 2016.
\newblock \href {https://doi.org/10.1016/j.tcs.2015.10.016} {\path{doi:10.1016/j.tcs.2015.10.016}}.

\bibitem{DBLP:journals/tcs/LeL19}
Ho{\`{a}}ng{-}Oanh Le and Van~Bang Le.
\newblock A complexity dichotomy for matching cut in (bipartite) graphs of fixed diameter.
\newblock {\em Theor. Comput. Sci.}, 770:69--78, 2019.
\newblock \href {https://doi.org/10.1016/j.tcs.2018.10.029} {\path{doi:10.1016/j.tcs.2018.10.029}}.

\bibitem{DBLP:conf/stacs/MarxOR10}
D{\'{a}}niel Marx, Barry O'Sullivan, and Igor Razgon.
\newblock Treewidth reduction for constrained separation and bipartization problems.
\newblock In Jean{-}Yves Marion and Thomas Schwentick, editors, {\em 27th International Symposium on Theoretical Aspects of Computer Science, {STACS} 2010, March 4-6, 2010, Nancy, France}, volume~5 of {\em LIPIcs}, pages 561--572. Schloss Dagstuhl - Leibniz-Zentrum f{\"{u}}r Informatik, 2010.
\newblock \href {https://doi.org/10.4230/LIPIcs.STACS.2010.2485} {\path{doi:10.4230/LIPIcs.STACS.2010.2485}}.

\bibitem{DBLP:conf/wg/PatrignaniP01}
Maurizio Patrignani and Maurizio Pizzonia.
\newblock The complexity of the matching-cut problem.
\newblock In Andreas Brandst{\"{a}}dt and Van~Bang Le, editors, {\em Graph-Theoretic Concepts in Computer Science, 27th International Workshop, {WG} 2001, Boltenhagen, Germany, June 14-16, 2001, Proceedings}, volume 2204 of {\em Lecture Notes in Computer Science}, pages 284--295. Springer, 2001.
\newblock \href {https://doi.org/10.1007/3-540-45477-2\_26} {\path{doi:10.1007/3-540-45477-2\_26}}.

\bibitem{stoerWagner}
Mechthild Stoer and Frank Wagner.
\newblock A simple min-cut algorithm.
\newblock {\em J. ACM}, 44(4):585–591, jul 1997.
\newblock \href {https://doi.org/10.1145/263867.263872} {\path{doi:10.1145/263867.263872}}.

\end{thebibliography}
\end{document}